%% file: ms.tex
\documentclass[runningheads]{llncs}
\usepackage{amssymb}
\usepackage{amsmath}
\usepackage{comment}
\usepackage{todonotes}
\usepackage{inputenc}
\usepackage{shuffle}
\usepackage{multirow}
\usepackage{listings}
\usepackage{mathabx}
\usepackage{hyperref}

\usepackage{diagbox}
\usepackage{slashbox}
\usepackage{tikz}
\usetikzlibrary{matrix}

\newcommand{\perm}{\operatorname{perm}}

\usepackage{apxproof}
\theoremstyle{plain}
\newtheoremrep{theorem}{Theorem}%[section]
\newtheoremrep{proposition}[theorem]{Proposition}
\newtheoremrep{lemma}[theorem]{Lemma}
\newtheoremrep{claim}[theorem]{Claim}
\newtheoremrep{conjecture}[theorem]{Conjecture}
\newtheoremrep{corollary}[theorem]{Corollary}
\newtheoremrep{definition}[theorem]{Definition}

\def\dotminus{\mathbin{\ooalign{\hss\raise1ex\hbox{.}\hss\cr
  \mathsurround=0pt$-$}}} 
 
\DeclareFontFamily{U}{bigshuffle}{}
\DeclareFontShape{U}{bigshuffle}{m}{n}{
  <5-8> s*[1.7] shuffle7
  <8->  s*[1.7] shuffle10
}{}
\DeclareSymbolFont{BigShuffle}{U}{bigshuffle}{m}{n}
\DeclareMathSymbol\bigshuffle{\mathop}{BigShuffle}{"001}
\DeclareMathSymbol\bigcshuffle{\mathop}{BigShuffle}{"002}

\newcommand{\pref}{\operatorname{Pref}}

\newcommand{\NL}{\textsf{NL}}

\newcommand{\PSPACE}{\textsf{PSPACE}}

\newcommand{\PTIME}{\textsf{P}}

\newenvironment{claiminproof}[1]{\medskip\par\noindent\underline{Claim:}\space#1}{}
\newenvironment{claimproof}[1]{\begin{quote}\par\noindent\emph{Proof of the Claim:}\space#1}{[\emph{End, Proof of the Claim}]\end{quote}}%\newline}
\begin{document}
\title{State Complexity of Permutation and Related Decision Problems on Alphabetical Pattern Constraints}
\titlerunning{Permutation on Alphabetical Pattern Constraints}

\author{Stefan Hoffmann\orcidID{0000-0002-7866-075X}}
\authorrunning{S. Hoffmann}
\institute{Informatikwissenschaften, FB IV, 
  Universit\"at Trier,  Universitätsring 15, 54296~Trier, Germany, 
  \email{hoffmanns@informatik.uni-trier.de}}
\maketitle              % typeset the header of the contribution
\begin{abstract}
  
  %The set of all strings Parikh equivalent to a string in a language $L$ is called the permutation of $L$.
  %Here, 
  We investigate the state complexity of the permutation operation,
  or the commutative closure, on Alphabetical Pattern Constraints (APC).
  This class corresponds to level $3/2$ of the Straubing-Th{\'e}rien hierarchy 
  and includes the finite, the piecewise-testable, or $\mathcal J$-trivial,
  and the $\mathcal R$-trivial and $\mathcal L$-trivial languages.
  We give a sharp state complexity bound %that is 
  expressed in terms of the longest strings in the unary projection
  languages of an associated finite language. This bound is already sharp for the subclass of finite languages.
  Additionally, for two
  subclasses, we give sharp bounds expressed in terms of the size of a recognizing
  input automaton and the size of the alphabet.
  Lastly, we investigate the inclusion and universality problem on APCs up to permutational equivalence. These two problems are known to be $\PSPACE$-complete on APCs in general, even for fixed alphabets.
  However, we show them to be decidable in polynomial time for fixed alphabets if
  we only want to solve them up to permutational equivalence.
  %This implies that
  %the original inclusion and universality problem, which are $\PSPACE$-complete
  %for general APCs, are polynomial time solvable for APCs describing
  %commutative languages.

  % todo parikh equiavlnce definieren
\keywords{state complexity \and finite automata \and alphabetic pattern constraint language \and commutative closure \and inclusion problem} 
\end{abstract}
%
%
%

% \begin{tabular}{c|c|c}
%      APC &  \\
%      $\perm()$ shuffle language & \\ 
%      General Chain Automata & \\ \hline 
%      $\perm(L)$-Universality & \\ 
%      $\perm(L)$-Inclusion &
% \end{tabular}

\section{Introduction}
\label{sec:introduction}

\input{introduction}

% todo noch die bemerkungen aus fin-perm paper übernehmen?

\section{State Complexity Bound of Permutation on APCs}

\input{apc}

\section{When $\perm(L^{\operatorname{simple}}(\mathcal A))$ is a Strict Shuffle Language}
\label{sec:shufflelang}

\input{shufflelang}

% todo dfa intersection auf diesen
% für chain automata an jedem übergang "kleinste" menge so dass übergänge nehmen.
% (für i-ten schritt)
\section{State Complexity %of Permutation 
on General Chain Automata} %  NFA}

\input{chain2}

\section{Complexity Results}
\label{sec:complexity}

\input{complexity}

\section{Conclusion}

 We have given a sharp upper bound for the number of states needed in a deterministic recognizing
 automata for the commutative closure of APCs.
 Additionally, we have shown that the recognizing automaton
 could be computed in polynomial time for fixed alphabets.
 Using this result, we have shown that the inclusion and universality problem
 modulo permutational equivalence are solvable in polynomial time
 for a fixed input alphabet. This contrasts with the general inclusion and universality problem
 for APCs. Both are $\PSPACE$-complete even for binary alphabets~\cite{DBLP:journals/iandc/KrotzschMT17}.
 For two subclasses of the APC languages, we have given sharp 
 bounds for the commutative closure expressed in the size of the input automata.
 In the case that the language is given by a general chain automaton,
 the result was only established for binary alphabets. The case for larger alphabets
 is still open.

\smallskip \noindent \footnotesize
\textbf{Acknowledgement.} I thank the anonymous reviewers for careful reading, noticing a reoccurring typo in the proof of Theorem~1 that was luckily spotted and fixed and
helping me identifying some unclear formulations throughout the text.

\bibliographystyle{splncs04}
\bibliography{ms} 
\end{document}

%% file: introduction.tex
In \emph{regular model checking}~\cite{DBLP:conf/concur/AbdullaJNS04}, a set of initial configurations is modelled as a regular language and the actions of the system are modeled as a 
rewriting relation.
%, which is a length-preserving rational relation
%over $\Sigma^* \times \Sigma^*$. 
For example, suppose we have an arbitrary number of processes that are connected linearly and need access to a common resource, but only one at a time and in order, starting from the first processor. Then, the state of a given processor
could be modeled by $\Sigma = \{0,1\}$, where $1$ means the processor has access
to the resource, and $0$ otherwise. The set of initial configurations is then the regular languages $10^*$, where a specific initial configuration is determined by the number of processors involved.
The transition relation is given by the rule $10 \mapsto 01$
and the set of reachable configurations is the language $0^*10^*$.
The bad configurations are given by the language $(0+1)^*1(0+1)^*1(0+1)^*$,
and we see that intersection of this set with the reachable configurations
is empty.
%As in this framework, the computation of the set of reachable configurations

The computation of the set of reachable configurations is the closure of the set of initial configurations under the rewriting relation. However, in this generality, the framework is Turing-complete and hence restrictions have to be imposed. In~\cite{DBLP:journals/iandc/BouajjaniMT07} the class
of \emph{Alphabetical Pattern Constraints} (APC) was introduced as a class
to describe initial and bad configurations, given by forbidden patterns,
that is closed
under semi-commutations. The constructions in~\cite{DBLP:journals/iandc/BouajjaniMT07} rely on an inductive transformation of an APC expression into another APC expression for the closure.
Here, our constructions will give a more direct and efficient procedure for the full commutative closure and will also yield deterministic automata, which we then use to devise polynomial time decision
procedures for the inclusion and universality problem up to permutational equivalence.

The state complexity of a regular language $L$
is the minimal number of states needed in a deterministic automaton
recognizing $L$. 
Investigating the state complexity of the result of a regularity-preserving operation on regular languages, depending
on the state complexity of the regular input languages,
was first initiated in~\cite{Mas70} and systematically started in~\cite{YuZhuangSalomaa1994}.
As the number of states of a recognizing automaton 
could be interpreted as the memory required to describe the recognized language
and is directly related to the runtime of algorithms employing regular languages, obtaining
state complexity bounds is a natural question with applications in verification, natural language
processing or software engineering~\cite{GaoMRY17}.

In general, the permutation operation is not regularity-preserving.
But it is regularity-preserving on finite languages, APCs and on group languages~\cite{DBLP:journals/iandc/BouajjaniMT07,DBLP:journals/iandc/GomezGP13,DBLP:conf/dcfs/Hoffmann20}.
The state complexity on group languages was studied in~\cite{DBLP:conf/dcfs/Hoffmann20},
but it is not known if the derived bounds are tight.
 The state complexity of the permutation operation
 on finite languages was first investigated
 in~\cite{DBLP:journals/tcs/ChoGHKPS17,DBLP:conf/dcfs/PalioudakisCGHK15}.
 However, sharp bounds were only obtained for subclasses and it is unknown
 if the general bound stated in~\cite{DBLP:journals/tcs/ChoGHKPS17,DBLP:conf/dcfs/PalioudakisCGHK15}
 is sharp. Surely, every finite language is an APC.

 The \emph{dot-depth hierarchy}~\cite{DBLP:journals/jcss/CohenB71}
 % todo noch zitieren?, knast infinte und übersicht} 
 %was introduced in~\cite{} as an 
 is an infinitely increasing hierarchy whose union is the class
 of star-free languages. This hierarchy was motivated by alternately
 increasing the combinatorial and sequential complexity of languages
 and corresponding recognizing devices~\cite{DBLP:journals/ita/Brzozowski76,DBLP:conf/birthday/Pin17a}.
 Later, the more fundamental \emph{Straubing-Th{\'e}rien hierarchy} 
 was introduced~\cite{DBLP:journals/mst/PlaceZ19,DBLP:journals/tcs/Straubing81,DBLP:journals/tcs/Therien81}.
 Here, we start with $\{ \emptyset, \Sigma^* \}$
 at level zero and, alternately, build (1) the half-levels: finite unions %polynomom oder marked product, todo??
 of \emph{marked products} of
 the form $L_0 a_1 L_1 a_1 \cdots a_k L_k$
 with $k \ge 0$, $a_1, \ldots, a_k \in \Sigma$
 and $L_1, \ldots, L_k$
 from the previous full-level
 or, (2) the full levels:
 the \emph{Boolean closure} of the previous half-level. % formal, wie in masopast "universality of confluent"-paper?
 More formally, 
 set $\mathcal L(0) = \{ \emptyset, \Sigma^* \}$
 and for $n \ge 0$, level $\mathcal L(n+\frac{1}{2})$
 consists of all finite unions of languages 
 $L_0 a_1 L_1 a_2 \cdots a_k L_k$ with $k \ge 0$,
 $L_0, \ldots, L_k \in \mathcal L(n)$
 and $a_1, \ldots, a_k \in \Sigma$, 
 and level $\mathcal L(n+1)$ consists of all finite Boolean
 combinations of languages from level $\mathcal L(n+\frac{1}{2})$.
 Every star-free language is contained in some level 
 of this hierarchy, which is also infinitely increasing.
 The different levels could also be characterized
 logically by the quantifier alternation of first order sentences~\cite{DBLP:journals/jcss/Thomas82}.

 The membership problem and related decision and separation problems
 with respect to the levels of both hierachies, and their connection 
 to logic, have sparked much interest~\cite{DBLP:journals/iandc/KrotzschMT17,DBLP:journals/mst/PlaceZ19,DBLP:journals/jcss/Thomas82}.
 The APCs precisely correspond to the languages of level $3/2$
 in the Straubing-Th{\'e}rien hierarchy~\cite{DBLP:journals/iandc/BouajjaniMT07,DBLP:journals/iandc/KrotzschMT17}.

 Green's relations are five equivalence relations, 
 named $\mathcal H$, $\mathcal R$, $\mathcal L$, $\mathcal J$ and $\mathcal D$,
 that characterize the elements of a semigroup in terms of the principal ideals
 they generate~\cite{Green51}.
 By the notion of the syntactic monoid~\cite{McNaughton71,DBLP:journals/tit/Schutzenberger56},
 these relations entered into formal language theory and
 proved to be useful in the classification of formal languages~\cite{DBLP:conf/lata/Colcombet11,Pin86,DBLP:reference/hfl/Pin97}.
 For example, it turned out that the $\mathcal J$-trivial, or piecewise-testable languages,
 are precisely the languages of level one in the Straubing-Th{\'e}rien hierarchy~\cite{DBLP:reference/hfl/Pin97,Simon75}.
 The $\mathcal H$-trivial languages are precisely the star-free languages~\cite{Schutzenberger65a}.
 Also, the $\mathcal R$-trivial and the $\mathcal L$-trivial languages
 are properly contained in level $3/2$ of the Straubing-Th{\'e}rien hierarchy, i.e., are APCs~\cite{DBLP:journals/iandc/BouajjaniMT07,BrFi80,DBLP:journals/iandc/KrotzschMT17}.

 % pin varieties + handbook zitieren
   % greens relations, charakterisierung pto, nfa
  % beziehung, für r-tivial j-trivial simon browzoswki + pin für übersicht zutieren (und eilenberg?)

\section{Preliminaries and Definitions} % todo eigene datei
\label{sec:preliminaries}

We assume the reader to have some basic knowledge of automata and complexity theory. For all unexplained notions, as, for example, regular expressions, the Nerode equivalence relation and 
%the 
more formal definitions of %common complexity classes
%like 
$\PSPACE$, the class of problems solvable with polynomially bounded space,
and $\PTIME$, the class of problems solvable in polynomial time,
we refer the reader to~\cite{HopUll79}.

For an alphabet (finite nonempty set) $\Sigma$, denote by $\Sigma^*$
the set of all finite words over the alphabet $\Sigma$ including
the \emph{empty word} $\varepsilon$.
If $u \in \Sigma^*$ and $a \in \Sigma$,
by $|u|$ we denote the \emph{length}
of $u$ and by $|u|_a$ the number of occurrences of
the symbol $a$ in $u$. 
%We also define the one-letter projection mappings
%$\pi_a : \Sigma^* \to \{a\}^*$ by $\pi_a(u) = a^{|u|_a}$
%for $u \in \Sigma^*$ and $a \in \Sigma$. For $L \subseteq \Sigma^*$ and $a \in \Sigma$,
%we set $\pi_a(L) = \{ \pi_a(u) : u \in L \}$.
%A word $u \in \Sigma^*$
%is a \emph{prefix} of a word $v \in \Sigma^*$
%if there exists $w \in \Sigma^*$
%such that $v = uw$. 
A \emph{language} over $\Sigma$
is any subset of $\Sigma^*$. Let $L \subseteq \Sigma^*$. We set $\pref(L) = \{ u \in \Sigma^* \mid \exists v \in \Sigma^* : uv \in L \}$.
A word $u \in \Sigma^*$
is a \emph{prefix} of a word $v \in \Sigma^*$, if $u \in \pref(\{v\})$.
For $a \in \Sigma$, the \emph{one-letter projection language} is
$\pi_a(L) = \{ a^{|u|_a} : u \in L \}$ and, for $u \in \Sigma^*$, we set $\pi_a(u) = a^{|u|_a}$.

For a natural number $n \ge 0$, we set $[n] = \{0,\ldots,n-1\}$.
For a finite subset $A$ of natural numbers, by $\max A$ and $\min A$ 
we denote the maximal and minimal element in $A$
 with respect to the usual order, where we set $\max\emptyset = \min\emptyset = 0$.

A \emph{nondeterministic finite automaton (NFA)} is given by % a quintuple
$\mathcal A = (\Sigma,Q,\delta, q_0, F)$,
where $\Sigma$ is an \emph{input alphabet},
$Q$ a finite set of \emph{states}, $\delta : Q \times \Sigma \to 2^Q$
the \emph{transition function}, having a set of states as image,
$q_0$ the \emph{initial state} and $F \subseteq Q$
the set of \emph{final states}.
If, for any $q \in Q$ and $a \in \Sigma$, we
have $|\delta(q, a)| \le 1$, then $\mathcal A$
is called a \emph{partial deterministic finite automaton (PDFA)}.
If $\mathcal A$ is a PDFA, then the transition
function is often written as a partial function $Q \times \Sigma \to Q$.
In the usual way, the transition function $\delta$ 
can be extended to the domain $Q \times \Sigma^*$.
The \emph{language recognized} by $\mathcal A$
is $L(\mathcal A) = \{ u \in \Sigma^* \mid \delta(q_0, u) \cap F \ne \emptyset \}$.
The finite \emph{simple language} 
associated
with $\mathcal A$ is
% diese darstellung für L^simple nur für ptoDFAs
%$L^{\operatorname{simple}}(\mathcal A) = \{ w \in L(\mathcal A) \mid \delta(q_0, u) \ne \delta(q_0, v) \mbox{ for distinct prefixes $u,v$ of $w$} \}$ 
$L^{\operatorname{simple}}(\mathcal A)
  = \{ w \in \Sigma^* \mid \text{$w$ labels
  a simple accepting path in $\mathcal A$} \}$,
 where a path is \emph{simple} if no state occurs more than once along the path, i.e.,
 the states we end up after each prefix (along the path) are distinct for distinct
 prefixes, and a path is \emph{accepting}
 if it starts at the initial state of $\mathcal A$
 and ends in a final state.
% $L^{\operatorname{simple}}(\mathcal A) = \{ w \in L(\mathcal A) \mid \delta(q_0,w) \setminus \left(\bigcup_{u \in \pref(w) \setminus \{w\}} \delta(q_0, u) \right) \ne \emptyset \}$ % so ist eps immer drin, wenn in L(A)
% associated
% with $\mathcal A$.
% labeling paths without loops -> simple paths
The language $L^{\operatorname{simple}}(\mathcal A)$ is the set of all words in $L(\mathcal A)$ that label paths with no loops\footnote{The
length of the longest word in $L^{\operatorname{simple}}(\mathcal A)$
is called the \emph{depth} in~\cite{DBLP:journals/corr/abs-1907-13115}.}.

 \begin{lemmarep}
 \label{lem:max_length_letter}
  Let $a \in \Sigma$
  and $n =  \max\{ |u|_{a} \mid u \in L^{\operatorname{simple}}(\mathcal A) \} + 1$.
  Then, for any $w \in \Sigma^*$ with\footnote{The assumption $|w|_{a} \ge n$
  is needed. For example, consider $L = a^{n-1}b^*$.} $|w|_{a} \ge n$, we have:
  $
   w \in \perm(L) \Leftrightarrow  wa \in \perm(L).
  $
    \end{lemmarep}
    \begin{proof}
      If $w \in \perm(L)$ is such that $|w|_{a_j} \ge n_j$,
      then in $\mathcal A$ we have to enter a loop labeled by $a_j$ when
      reading $w$, as after removing the loops every 
      accepting path has at most $n_j - 1$ transitions labeled by the letter $a_j$.
      So, we can read in $a_j$ one additional time and find $wa_j \in \perm(L)$.
      Conversely, if $wa_j \in \perm(L)$, with the same reasong, we find
      that we have to traverse a loop labeled by $a_j$ at least once,
      which could be left out and so $w \in \perm(L)$.
    \end{proof}

The \emph{state complexity} of a regular language is the smallest number of states 
in any PDFA recognizing the language.

Let $\mathcal A = (\Sigma, Q, \delta, q_0, F)$.
A state $q \in Q$ is said
to be \emph{reachable} from a state $p \in Q$, 
if there exists $u \in \Sigma^*$
such that $q \in \delta(p, u)$.

An automaton is called \emph{partially ordered}, if the reachability 
relation is a partial order. Equivalently, if the only loops are self-loops. Partially ordered automata are also known as \emph{weakly acyclic} automata~\cite{Ryzhikov19a}.

The \emph{shuffle operation} of two languages $U, V \subseteq \Sigma^*$ is defined by
\begin{multline*}
 U \shuffle V  = \{ w \in \Sigma^*  \mid  w = x_1 y_1 x_2 y_2 \cdots x_n y_n 
    \mbox{ for some words } \\ x_1, \ldots, x_n, y_1, \ldots, y_n \in \Sigma^* 
    \mbox{ such that } x_1 x_2 \cdots x_n \in U \mbox{ and } y_1 y_2 \cdots y_n \in V\}.
\end{multline*}
and $u \shuffle v = \{u\}\shuffle \{v\}$ for $u,v \in \Sigma^*$.
For languages $L_1, \ldots, L_n \subseteq \Sigma^*$, we set $\bigshuffle_{i=1}^n L_i = L_1 \shuffle \ldots \shuffle L_n$.
Let $L \subseteq \Sigma^*$. If $L = \bigshuffle_{a \in \Sigma} \{ a^{|u|_a} \mid u \in L \}$, then we call it a \emph{strict shuffle language}. 

\begin{example} Let $\Sigma = \{a,b\}$.
\begin{enumerate}
\item If $u \in \Sigma^*$, then $\perm(u)$ is a strict shuffle language.
\item The language $\{ u \in \{a,b\}^* \mid |u|_a = 1 \mbox{ and } 2 \le |u| \le n \} = \{ a\} \shuffle \{ b, b^2, \ldots, b^{n-1} \}$
 is a strict shuffle language.
\item $\perm(\{aabb, ab\})$ is not a strict shuffle language.
\item $\perm(\{ aaabbb, abbb, aaab, ab \})$ is a strict shuffle language.
\end{enumerate}
\end{example}

The \emph{permutation operation}, or \emph{commutative closure},
on a language is the set of words that we get when permuting the letters
of the words from the language. Formally, for $L \subseteq \Sigma^*$,
we set $\perm(L) = \{ u \in \Sigma^* \mid \exists v \in L \ \forall a \in \Sigma : |u|_a = |v|_a \}$. For example, $\perm(\{abb\}) = \{ abb, bab, bba \}$.
For $u \in \Sigma^*$, we also write $\perm(u)$ for $\perm(\{u\})$.
A language $L \subseteq \Sigma^*$
is called \emph{commutative}, if $\perm(L) = L$.
Note that for strict shuffle languages
$L \subseteq \Sigma^*$ we have $\perm(L) = L$.
% for pto aut language when all self-loops removed, i.e., the DAG resulting.

An \emph{Alphabetical Pattern Constraint} (APC)
is an expression\footnote{With the shorthand $\Gamma^* = (a_1 + \ldots + a_n)^*$
for $\Gamma = \{a_1, \ldots, a_n\} \subseteq \Sigma$.}
$p_1 + \ldots + p_n$,
where each $p_i$ is of the form $\Sigma_0^* a_1 \Sigma_1^* \cdots a_n \Sigma_n^*$
with $\Sigma_0, \ldots, \Sigma_n \subseteq \Sigma$
and $a_1, \ldots, a_n \in \Sigma$.
In the following, we will \emph{not distinguish between the expression
and the language} it denotes, and taking the liberty to denote ``$+$'' by the union symbol as well.
Hence, an APC is a finite union of languages of the form 
$\Sigma_0^* a_1 \Sigma_1^* \cdots a_n \Sigma_n^*$ as above. Equivalently,
as concatenation distributes over union, it is the closure
of the subsets $\Gamma^*$, $\Gamma \subseteq \Sigma$, and $\{a\}$ for $a \in \Sigma$
under concatenation and finite union\footnote{Note that, for $\Sigma_0, \Sigma_1 \subseteq \Sigma$
nonempty, we have $\Sigma_0^* \Sigma_1^* = \Sigma_0^* \cup \bigcup_{a \in \Sigma_1} \Sigma_0^* a \Sigma_1^*$.}.
The APCs are precisely the languages recognized
by partially ordered NFAs~\cite{DBLP:journals/iandc/KrotzschMT17,SchwentickTV01}.

%% file: apc.tex
The APCs are closed under the permutation operation. However, they are not closed under complementation.
For example, the complement of $\Sigma^*aa\Sigma^* \cup \Sigma^*bb\Sigma^* \cup b\Sigma^* \cup \Sigma^*a$
over $\Sigma = \{a,b\}$ is $(ab)^*$.
As $\perm((ab)^*) = \{ u \in \Sigma^* \mid |u|_a = |u|_b \}$ is not regular, it is not an APC.
This also shows that level $3/2$
of the Straubing-Th{\'e}rien hierarchy is the lowest level in which 
the permuation of any language is regular.

\begin{remark}
 Let $L = \bigcup_{i=1}^m \Sigma_0^{(i)} a_1^{(i)} \Sigma_1^{(i)}\cdots a_{n_i}^{(i)} \Sigma_{n_i}^{(i)}$.
 Set $\Gamma^{(i)} = \Sigma_0^{(i)} \cup \ldots \cup \Sigma_{n_i}^{(i)}$.
 Then,
 $ 
  \perm(L) 
  =\bigcup_{i=1}^m \perm(a_1^{(i)} \cdots a_{n_i}^{(i)} \Gamma^{(i)} )
  = \bigcup_{i=1}^m \perm(a_1^{(i)} \cdots a_{n_i}^{(i)}) \shuffle \Gamma^{(i)}. 
 $
 %If $L$ is given by a partially ordered NFA $\mathcal A$, then it could be shown
 %that $\perm(L)$
 %equals 
 % L^simple als DAG
 Hence, the permutational closure, as
 a finite union of languages of the form $\Gamma^* a_1 \Gamma^* \ldots a_n \Gamma^*$,
 is itself an APC.
\end{remark}

\begin{theorem}
\label{thm:apc_bound}
 Let $L$ be an APC recognized by a partially ordered NFA $\mathcal A$.
 Then, $\perm(L)$
 is recognizable by a PDFA that uses at most
 (where we set $\max\emptyset = 0$)
 %https://math.stackexchange.com/questions/1868878/inf-e-sup-of-empty-set
 %
 % in dem zusammenhang wäre sup besser, da
 % max eigentlich element der menge sein muss
 %\footnote{We set $\max\emptyset = 0$.}
 \[ % simple in dem fall als dag sehen (fussnote)
  \prod_{a \in \Sigma} ( \max\{ |u|_a : u \in L^{\operatorname{simple}}(\mathcal A) \} + 1 )
 \]
  many states and this bound is sharp even for finite languages.
\end{theorem} 
\begin{proof}
 Suppose we have $k$ symbols
 and $\Sigma = \{a_1, \ldots, a_k \}$.
 Set $n_j = \max\{ |u|_{a_j} \mid u \in L^{\operatorname{simple}}(\mathcal A) \} + 1$ for $j \in \{1,\ldots, k\}$.
     Construct $\mathcal B = (\Sigma, Q, \delta, q_0, F)$
     with $Q = [n_1+1] \times \ldots \times [n_k+1]$
     and
     \begin{multline*}
      \delta((s_1, \ldots, s_k), a_j) 
       = \\ \left\{ \begin{array}{ll}
          (s_1, \ldots, s_{j-1}, s_j + 1, s_{j+1}, \ldots, s_k) & \mbox{if } s_j < n_j ;\\ 
          (s_1, \ldots, s_k)                        & \mbox{if } s_j = n_j \mbox{ and } a_1^{s_1} \cdots a_k^{s_k}a_j \in \perm(\pref(L)).
          \end{array}\right.
     \end{multline*}
    Also $q_0 = (0,\ldots, 0)$ and
     $
      F = \{ \delta(q_0, w) \mid w \in L \mbox{ and } \forall j \in \{1,\ldots,k\} : |w|_{a_j} \le n_j \}. 
     $
    
    \begin{claiminproof}
     We have $L(\mathcal B) = \perm(L)$.
    \end{claiminproof} 
    \begin{claimproof}
      By Lemma~\ref{lem:max_length_letter}, for any $w \in \Sigma^*$ with $|w|_{a_j} \ge n_j$, we have 
  \[
   w \in \perm(L) \Leftrightarrow  wa_j \in \perm(L).
  \]
  Let $w \in \perm(L)$.
  Then $a_1^{\min\{ n_1, |w|_{a_1} \}} \cdots a_k^{\min\{ n_k, |w|_{a_k} \}} \in \perm(L)$.
  Hence, % todo wegen def F
  \[ 
   \delta(q_0, a_1^{\min\{ n_1, |w|_{a_1} \}} \cdots a_k^{\min\{ n_k, |w|_{a_k} \}}) \in F.
  \]
  Furthermore, if $|w|_{a_j} \ge n_j$,
  then, for  \[ q =  (\min\{ n_1, |w|_{a_1} \},\ldots,\min\{ n_k, |w|_{a_k} \}), \]
  we have
  $
   \delta(q, a_j) = q
  $
  for any $j \in \{1,\ldots,k\}$ such that $|w|_{a_j} \ge n_j$.
  So, \[ % todo hier genauer, wegen b^*a zum beispiel b^Na eingelesen.
  \delta(q_0, w) = \delta(q_0, a_1^{\min\{ n_1, |w|_{a_1} \}} \cdots a_k^{\min\{ n_k, |w|_{a_k} \}}) \in F.
  \]
  
  Conversely, suppose $\delta(q_0, w) \in F$.
  If $|w|_{a_j} > n_j$ with $j \in \{1,\ldots,k\}$, then, for  the state  $q =  (\min\{ n_1, |w|_{a_1} \},\ldots,\min\{ n_k, |w|_{a_k} \})$,
  we have
  $
   \delta(q, a_j) = q
  $
  and so $\perm(a_1^{\min\{ n_1, |w|_{a_1} \}} \cdots a_k^{\min\{ n_k, |w|_{a_k} \}}a_j) \subseteq \perm(L)$
  by the above definition of the transition function $\delta$.
  Hence, the letter $a_j$ could be appended $n_j - |w|_{a_j}$ many times and $\mathcal B$ stays in the same
  state, for every such letter with $|w|_{a_j} > n_j$.
  So, we find $ \perm(a_1^{|w|_{a_1}} \cdots a_k^{|w|_{a_k}}) \subseteq \perm(L), $
  which is equivalent to $w \in \perm(L)$.
  \end{claimproof}
  That the bound is sharp is shown
 in Remark~\ref{rem:lower_bound_shuffle_lang}.~\qed
\end{proof} 

\begin{lemmarep}
\label{lem:relation_sc_simple}
 Let $\mathcal A = (\Sigma, Q, \delta, q_0, F)$
 be a partially ordered NFA.
 If any NFA for $L^{\operatorname{simple}}(\mathcal A)$
 needs at least $n$ states, then $|Q| \ge n$.
 A similar statement holds true for PDFAs.
\end{lemmarep}
\begin{proof}
 By deleting all self-loops in $\mathcal A$, we 
 find an automaton for $L^{\operatorname{simple}}(\mathcal A)$
 with $|Q|$ many states. Similarly if $\mathcal A$
 is a PDFA.~\qed
\end{proof}

Let $\mathcal A = (\Sigma, Q, \delta, q_0, F)$ be a partially ordered NFA.
As $L^{\operatorname{simple}}(\mathcal A)$ is finite, every path from the start state
to a final state in any recognizing automaton has no loops. Hence, the length of a longest string in $L^{\operatorname{simple}}(\mathcal A)$
is a lower bound for the number of states of any NFA
recognizing $L^{\operatorname{simple}}(\mathcal A)$. 
Surely, for $a \in \Sigma$ and $u \in L^{\operatorname{simple}}(\mathcal A)$, the number
$|u|_a$ is a lower bound for the length of the longest string in $L^{\operatorname{simple}}(\mathcal A)$.
So, combining with Lemma~\ref{lem:relation_sc_simple},
%we have $|u|_a \le |Q|$ for $u \in L^{\operatorname{simple}}(\mathcal A)$ and $a \in \Sigma$.
we have $\max\{ |u|_a : u \in L^{\operatorname{simple}}(\mathcal A) \} \le |Q|$ for $a \in \Sigma$.
%as removing the self-loops in $\mathcal A$ yields
%a recognizing automaton for $L^{\operatorname{simple}}(\mathcal A)$. 
This yields the next corollary to Theorem~\ref{thm:apc_bound}.

\begin{corollary}
 Let $L$ be an APC recognized by a partially ordered NFA with $n$ states.
 Then, $\perm(L)$ is recognizable by a PDFA with at most $n^{|\Sigma|}$ many states.
\end{corollary}

We have formulated Theorem~\ref{thm:apc_bound} and the above corollary
in terms of partially ordered NFAs recognizing a given APC.
However, APC expressions and partially ordered NFAs are closely connected, for example, see
Lemma~\ref{lem:NFA_to_APC_P} in Section~\ref{sec:complexity}. Hence a corresponding statement could be made
for APC expressions, where $L^{\operatorname{simple}}(\mathcal A)$ corresponds
to the set of words resulting if we delete all subexpressions $\Sigma_i^{*}$
in the parts of the unions.

% bzw geht sogar, wenn shuffle, weil dann max-summe <= longest word <= states?
% bemekrung für chani geht argument wie für einzelnes "wort" nicht, beispiel sigma^n
%
% comm closure aperidoic [spezialfall von anderem] und automat in P [alphabet fix] konstruierbar O(n^k)
%
% subkapitel zu APC languages und Straubing-Th-hierarchy

%% file: shufflelang.tex
% stimmt nicht.
%Here, we investigate a class of languages
%whose members attain the bound from 
%Theorem~\ref{thm:apc_bound}, implying that this bound is sharp.

Here, we investigate a class of languages
for which we can devise a sharp bound expressed in the size of the input NFA.
The bound is formulated with the number
of states and the size of the alphabet of the input automaton. As the bound
is sharp for a subclass of languages, it also yields a lower bound for the general case.
%First, a few examples of strict shuffle languages, defined in Section~\ref{sec:preliminaries}.

%
% Prop L = U_1 shuffle, dann kleinster PDFA produkt, wobei n_J größte kleinstens PDFa U_j
% 
% Lemma kleinster PDFA mindestens summe max zustände

\begin{comment} % ne gilt wohl nicht, da Prop 8 aus altem paper falsch
\begin{corollary}
 Let $\mathcal A$ be a partially ordered NFA.
 If $\perm(L^{\operatorname{simple}}(\mathcal A))$
 is a strict shuffle language, then the bound
 of Theorem~\ref{thm:apc_bound} is attained, i.e.,
 any PDFA for $\perm(L(\mathcal A))$
 needs 
 \[
 \]
 many states.
\end{corollary}
\begin{proof}
 %Suppose $\Sigma = \{a_1, \ldots, a_k\}$ and write $\perm(L^{\operatorname{simple}}(\mathcal A)) = \bigshuffle_{j=1}^k U_j$
 %with $U_j \subseteq \{a_j\}^*$.
 Write $\perm(L^{\operatorname{simple}}(\mathcal A)) = \bigshuffle_{a \in \Sigma} U_a$
 with $U_a \subseteq \{a\}^*$.
 The language $\perm(L^{\operatorname{simple}}(\mathcal A))$
 is finite. So, every $U_a$
 is finite and the minimal PDFA
 for each $U_a$, $a \in \Sigma$, 
 has precisely $|U|_a + 1$
 many states.
 By Proposition~\ref{todo}, 
 a minimal PDFA for
 $\perm(L^{\operatorname{simple}}(\mathcal A))$
 needs $\prod_{a \in \Sigma} |U|_a$ % vorsicht, final alle zusammenfassbar...
 for 
 Also,
 \[
  |U|_a = \max\{ \perm(L^{simple}) \} = \max\{ L^{\simple} \}. 
 \]
\end{proof}
\end{comment}

 For finite strict shuffle languages, we can derive the following lower bound
 for the size of recognizing NFAs, which we will need in the proof of Theorem~\ref{thm:shuffle_lang_bound}.
 
 \begin{lemmarep} 
 \label{lem:sc_lower_bound}
  Let $L\subseteq \Sigma^*$ be finite.
  If $\perm(L)$ is a strict shuffle language,
  then any NFA recognizing $L$ % for $L$
  needs at least
  $
   ( \sum_{a \in \Sigma} \max\{ |u| : u \in \pi_a(L) \} ) + 1
  $ 
  many states. 
 \end{lemmarep}
 \begin{proof}
  Suppose $\Sigma = \{a_1, \ldots, a_k\}$.
  Set $m_i = \max \{ |u| \mid u \in \pi_{a_i}(L)\}$ for $i \in \{1,\ldots, k\}$ and $w = a_1^{m_1} \cdots a_k^{m_k}$.
  Then $w \in \pi_{a_1}(L) \shuffle \ldots \shuffle \pi_{a_k}(L) = \perm(L)$.
  Hence, we find $u \in L$ with $u \in \perm(w)$.
  
  If $v \in L$ is arbitrary, then $|v| = \sum_{i=1}^k |v|_{a_i} \le \sum_{i=1}^k m_i = |u|$.
  So, $|u| = \max \{ |w'| \mid w' \in L \}$.
  For finite languages, as we could not have any loops on a path from the start
  state to some final state, the length of the longest string in the language plus one
  is a lower bound for the number of states for any recognizing NFA.
  As $|u|$ is a longest string in $L$, this gives
  the claim.~\qed
 \end{proof}

% längestes wort in L^simple <= sumem maximaler länge einzelteile <= stc(L)
 
 % https://www.wolframalpha.com/input/?i=maximize+%28x%2B1%29%28y%2B1%29%28z%2B1%29%2C+x%2By%2Bz%2B1%3C%3D101%2C0%3C%3Dx%2C0%3C%3Dy%2C0%3C%3Dz
 
 Next, we state the main result of this section.
 
\begin{theorem}
\label{thm:shuffle_lang_bound}
 Let $L$ be an APC language recognized by a partially ordered NFA~$\mathcal A$ with $n$ states
 such that $\perm(L^{\operatorname{simple}}(\mathcal A))$
 is a strict shuffle language.
 Then, $\perm(L)$ is recognizable by a PDFA with at most
 \[
   \left\lceil\frac{n-1}{|\Sigma|}  + 1 \right\rceil^{|\Sigma|}
 \]
 many states and this bound is sharp even for finite languages.
\end{theorem}
\begin{proof}
 By Lemma~\ref{lem:sc_lower_bound},
 any automaton for $L^{\operatorname{simple}}(\mathcal A)$
 needs at least $(\Sigma_{a \in \Sigma} \max\{ |u|_a : u \in \pi_a(L) \}) + 1$
 many states.
 So, by Lemma~\ref{lem:relation_sc_simple} we have $0 \le (\Sigma_{a \in \Sigma} \max\{ |u|_a : u \in \pi_a(L) \}) + 1 \le n$.
 The value
 $
   \prod_{a \in \Sigma} ( \max\{ |u|_a : u \in L^{\operatorname{simple}}(\mathcal A) \} + 1 ) 
 $
 from Theorem~\ref{thm:apc_bound} with the constraint 
 $0 \le ( \Sigma_{a \in \Sigma} \max\{ |u|_a : u \in \pi_a(L) \} ) + 1 \le n$
 is maximized\footnote{More precisely, if $\Sigma = \{a_1, \ldots, a_k\}$,
 we seek to maximize the function $f(x_1, \ldots, x_n) = \prod_{i=1}^k (x_i + 1)$
 due to the constraint $0 \le \sum_{i=1}^k x_i \le n - 1$, which happens for $x_1 = \ldots = x_k = \frac{n-1}{k}$
 with maximum value $\left(\frac{n-1}{k} + 1 \right)^k$.}
 if $\max\{ |u|_a : u \in L^{\operatorname{simple}}(\mathcal A) \}$ 
 equals $( n - 1 ) / |\Sigma|$ for every $a \in \Sigma$,
 which gives the claim.
 That the bound is sharp is shown
 in Remark~\ref{rem:lower_bound_shuffle_lang}.~\qed
\end{proof}

Note that for a single word $u \in \Sigma^*$,
we have $\perm(u) = \bigshuffle_{a \in \Sigma} \pi_a(u)$, i.e.,
the commutative closure is a strict shuffle language.
Hence, we get the next corollary from Theorem~\ref{thm:shuffle_lang_bound},
which is also sharp, as shown by Remark~\ref{rem:lower_bound_shuffle_lang}.
 
\begin{corollary}
\label{cor:shuffle_lang_bound}
 Let
 $
  L = \Sigma_0^* a_1 \Sigma_1^* a_2 \cdots a_m \Sigma_m^*.
 $
 Then, $\perm(L)$ is recognizable by a PDFA
 with at most $\lceil m/|\Sigma| + 1 \rceil^{|\Sigma|}$
 many states. In particular, the commutative closure
 of a single word $u$ could be recognized by a PDFA
 with at most $\left\lceil |u|/|\Sigma| + 1 \right\rceil^{|\Sigma|}$
 many states and this bound is sharp.
\end{corollary}
\begin{proof}
 The NFA
 $\mathcal A$ with state set $Q = \{ q_0, q_1, \ldots, q_m \}$,
 transition function $\delta(q_i, a) = \{ q_i : a \in \Sigma_i \} \cup \{ q_{i+1} : i < m \mbox{ and } a = a_{i+1} \}$
 for $i \in \{0,\ldots,m\}$ and $a \in \Sigma$, start state $q_0$ and final state set $\{q_m\}$ recognizes $L$.
 We have $L^{\operatorname{simple}}(\mathcal A) = \{ a_1 a_2 \cdots a_m \}$
 and $\perm(L^{\operatorname{simple}}(\mathcal A))$ is a strict shuffle language.
 Note that, by Lemma~\ref{lem:relation_sc_simple} and Lemma~\ref{lem:sc_lower_bound},
 $\mathcal A$ has the least possible number of states.
 Then, Theorem~\ref{thm:shuffle_lang_bound}
 gives the claim and the bound is sharp by Remark~\ref{rem:lower_bound_shuffle_lang}.~\qed
\end{proof}

\begin{remark}\label{rem:lower_bound_shuffle_lang}
 Suppose $\Sigma = \{a_1, \ldots, a_k\}$.
 Let $m > 0$ and $u = a_1^m \cdots a_k^m$.
 Then, any PDFA recognizing $\perm(u)$ needs at least $(m+1)^k$ many states.
 For let $0 \le m_i, n_i \le m$, $i \in \{1,\ldots, k\}$,
 such that there exists $j \in \{1,\ldots,k\}$
 with $m_j < n_j$.
 Then, choose $r_i$ for each $i \in \{1,\ldots,k\}$
 such that $n_i + r_i = m$.
 Set $w = a_1^{n_1} \cdots a_k^{n_k} a_1^{r_1}\cdots a_k^{r_k}$.
 As $|w|_{a_i} = m$ for any $i \in \{1,\ldots,k\}$,
 we find $w \in \perm(u)$.
 However, for $w' = a_1^{m_1} \cdots a_k^{m_k} a_1^{r_1}\cdots a_k^{r_k}$
 we have $|w'|_{a_j} < m$,
 so that $w'\notin \perm(u)$.
 So, $w$ and $w'$ represent different Nerode right-congruence classes~\cite{HopUll79}
 for the language $\perm(u)$, which yields the lower bound 
 for the number of states of any recognizing automaton.

 As $u$ is recognizable by a minimal NFA $\mathcal A$
 with $k\cdot m + 1$
 many states, $|u|_{a_i} = m$ for any $i \in \{1,\ldots,k\}$
 and $L^{\operatorname{simple}}(\mathcal A) = L(\mathcal A)$,
 as $L(\mathcal A)$ is finite,
 the bounds of Theorem~\ref{thm:apc_bound},
 Theorem~\ref{thm:shuffle_lang_bound}
 and of Corollary~\ref{cor:shuffle_lang_bound}
 are all meet by this example.
\end{remark}

%% file: chain2.tex
A \emph{general chain automaton} $\mathcal A = (\Sigma, Q, \delta, q_0, F)$
is a NFA such that the state set is totally
ordered, i.e., we can assume $Q = \{ 0, \ldots, n-1 \}$
with the usual order and $q_0 = 0$ and $F = \{n-1\}$ and, for any
$q \in Q \setminus\{n-1\}$ and $a \in \Sigma$,
we have $\delta(q, a) \subseteq \{q, q+1\}$.
If $\mathcal A$ is a general chain automaton, then 
$L^{\operatorname{simple}}(\mathcal A) \subseteq \Sigma^{n-1}$.

These automata, with no self-loops allowed\footnote{This is no restriction when we have no self-loops.}, 
were introduced in~\cite{DBLP:journals/tcs/ChoGHKPS17} under the name chain automata.
The sharp bound we will give
is essentially an adaption of the bound
derived in~\cite{DBLP:journals/tcs/ChoGHKPS17}. Note that we only have
a result for binary alphabets.
% beliebige in conclusion, auch schranke wie für shuffle
% im allgemeinen? todo

\begin{propositionrep}
\label{prop:bound_chain_automata}
 Let $\Sigma = \{a,b\}$ and $\mathcal A$ be a general chain automaton with $n$ states.
 Then, $\perm(L(\mathcal A))$
 is recognizable by a PDFA with at most $\frac{n^2 + n + 1}{3}$
 many states and this bound is sharp even on finite languages. %sharpness in paper.
\end{propositionrep}
\begin{proofsketch}
Let the set of states of $\mathcal A$ be $\{0, \ldots, n-1\}$,
 where $0$ is the start state and $n-1$ is the only final state.
 Set $\Gamma = \{ x \in \Sigma \mid \exists q \in Q : q \in \delta(q, x) \}$,
 the symbols which label self-loops.
 Note that $L(\mathcal A)$ is finite if and only if $\Gamma = \emptyset$.
 For $0 \le h \le n-2$, the transitions only go
 from $h$ to $h+1$ or we have a self-loop from $h$ to $h$.
 We have three possibilities 
 for outgoing transitions from a state $0 \le h \le n - 2$
 that are not self-loops:
 \begin{enumerate}
 \item $\{ h + 1 \} \subseteq \delta(h,a)$ and $\delta(h, b) \cap \{h + 1\} = \emptyset$  ($a$-transition);
 
 \item $\{ h + 1\} \subseteq \delta(h,b)$ and $\delta(h, a) \cap \{h+1\} = \emptyset$ ($b$-transition); 
 
 \item $\{ h + 1 \} \subseteq \delta(h,a) \cap \delta(h,b)$ ($a \& b$-transition).
 \end{enumerate}

 The order of the different types of transitions ($a$, $b$, or $a \& b$)
 of $\mathcal A$ does not affect the language $\perm(L(\mathcal A))$.
 A similar reasoning applies to the self-loops.
 Hence, without loss of generality, we can assume that $\mathcal A$
 has first a (possibly empty) sequence of $a$-transitions,
 followed by a (possibly empty) sequence of $b$-transitions,
 followed by a (possibly empty) sequence of $a\&b$-transitions
 and only self-loops with labels from the (possibly empty)
 subset $\Gamma \subseteq \Sigma$ at the final state.
 Thus, we can assume that $L(\mathcal A) = a^ib^j(a+b)^k\Gamma^*$
 for some non-negative integers $i,j,k$
 such that $i + j + k = n - 1$.
 By modifiying a construction from~\cite{DBLP:journals/tcs/ChoGHKPS17},
 we can construct a PDFA
 for $\perm(L(\mathcal A)$
 with $f(i,j,k) = (i+1)\cdot (j+1) + k\cdot j + k\cdot i + k
 $
 many states.
 In order to get an upper bound for the state complexity of $\perm(L(\mathcal A))$
 as a function of the size of $\mathcal A$,
 we determine for which values of $i,j,k$, where
 $i + j + k = n - 1$, the function $f(i,j,k)$ has a maximal value.
 The function $f$ is maximized if $ij + kj + ki$ is maximal,
 thus if $i = j = k = \frac{n-1}{3}$. More generally,
 \[
  \max_{i+j+k=n-1} f(i,j,k) = \left\{ % todo prüfen
   \begin{array}{ll}
  \frac{n^2 + n + 1}{3} & \mbox{if } n \equiv 1 \pmod{3}; \\
  \frac{n^2 + n}{3}     & \mbox{otherwise.}
  \end{array}
  \right.
 \]
 In~\cite[Lemma 4.2]{DBLP:journals/tcs/ChoGHKPS17}, as every chain automaton
 is a general chain automaton recognizing a finite language,
 it was shown that for $n \equiv 1 \pmod{3}$ 
 there exists a language recognized by a chain automaton
 with $n$ states such that any automaton
 for the commutative closure needs at least $\frac{n^2 + n + 1}{3}$
 many states.~\qed
\end{proofsketch}
\begin{proof}
 Let the set of states of $\mathcal A$ be $\{0, \ldots, n-1\}$,
 where $0$ is the start state and $n-1$ is the only final state.
 Set $\Gamma = \{ x \in \Sigma \mid \exists q \in Q : q \in \delta(q, x) \}$,
 the symbols which label self-loops.
 Note that $L(\mathcal A)$ is finite if and only if $\Gamma = \emptyset$.
 For $0 \le h \le n-2$, the transitions only go
 from $h$ to $h+1$ or we have a self-loop from $h$ to $h$.
 We have three possibilities 
 for outgoing transitions from a state $0 \le h \le n - 2$
 that are not self-loops:
 \begin{enumerate}
 \item $\{ h + 1 \} \subseteq \delta(h,a)$ and $\delta(h, b) \cap \{h + 1\} = \emptyset$  ($a$-transition);
 
 \item $\{ h + 1\} \subseteq \delta(h,b)$ and $\delta(h, a) \cap \{h+1\} = \emptyset$ ($b$-transition); 
 
 \item $\{ h + 1 \} \subseteq \delta(h,a) \cap \delta(h,b)$ ($a \& b$-transition).
 \end{enumerate}

 The order of the different types of transitions ($a$, $b$, or $a \& b$)
 of $\mathcal A$ does not affect the language $\perm(L(\mathcal A))$.
 A similar reasoning applies to the self-loops.
 Hence, without loss of generality, we can assume that $\mathcal A$
 has first a (possibly empty) sequence of $a$-transitions,
 followed by a (possibly empty) sequence of $b$-transitions,
 followed by a (possibly empty) sequence of $a\&b$-transitions
 and only self-loops with labels from the (possibly empty)
 subset $\Gamma \subseteq \Sigma$ at the final state.
 Thus, we can assume that $L(\mathcal A) = a^ib^j(a+b)^k\Gamma^*$
 for some non-negative integers $i,j,k$
 such that $i + j + k = n - 1$.
 Now, the language $\perm(L(\mathcal A))$ is recognized by the PDFA
 $\mathcal B_{i,j,k}
  = (\{a,b\}, \gamma, q_0, F_B)$ where
 \begin{multline*}
  Q = \{ (r,s) \mid 0 \le r \le i + k, 0 \le s < j \} \cup \\ 
      \{ (r,s) \mid 0 \le s \le j + k , 0 \le r < i \} \cup \{z_0,z_1,\ldots,z_k\},
 \end{multline*}
 $F_B = \{ z_k \}$, $q_0 = (0,0)$ and the transitions are defined
 by setting\footnote{The automaton is essentially the
 same as given in~\cite{DBLP:journals/tcs/ChoGHKPS17}
 for chain automata except for adding certain self-loops
 for symbols in $\Gamma$. The proof is also very similar.}, 
 for $(r, s) \in Q$,
 \begin{align*}
     \gamma((r,s), a) & = \left\{
     \begin{array}{ll}
     (r+1,s) & \mbox{if } r < i - 2 \mbox{ or } (r < i + k \mbox{ and } s < j); \\
     z_{s-j} & \mbox{if } r = i - 1 \mbox{ and } s \ge j; \\
     (r,s)   & \mbox{if } r + 1 > i + k \mbox{ and } a \in \Gamma; \\
     \mbox{undefined} & \mbox{if } r + 1 > i + k \mbox{ and } a \notin \Gamma;
     \end{array}
     \right. \\
     \gamma((r,s),b) & = \left\{
     \begin{array}{ll}
     (r, s+ 1) & \mbox{if } s < j - 2 \mbox{ or } (s < j + k \mbox{ and } r < i); \\
     z_{r-i}   & \mbox{if } s = j - 1 \mbox{ and } r \ge i; \\
     (r,s)     & \mbox{if } s + 1 > j + k \mbox{ and } b \in \Gamma; \\
     \mbox{undefined} & \mbox{if } s + 1 > j + k \mbox{ and } b \notin \Gamma; 
     \end{array}
     \right.
 \end{align*}
 and, for $x \in \Sigma$, 
 \[
  \gamma(z_l, x) = \left\{ 
  \begin{array}{ll}
   z_{l+1} & \mbox{if } 0 \le l < k; \\
   z_k     & \mbox{if } l = k \mbox{ and } x \in \Gamma; \\
   \mbox{undefined} & \mbox{if } l = k \mbox{ and } x \notin \Gamma.
  \end{array}
  \right.
 \]
 % todo, auch ein bild?
 A computation of $\mathcal B_{i,j,k}$ reaches a state
 of the form $(r,s)$ after encountering
 $r$ occurrences of $a$ if $r < i + k$ and $s < j$
 or at least $i + k$ occurrences of $a$ if $r = i + k$ and $s < j$
 and $s$ occurrences
 of $b$ if $s < j + k$ and $r < i$
 or at least $j + k$ occurrences of $b$ if $s = j + k$ and $r < i$.
 A state $z_l$, $0 \le l \le k$, is reached after encountering
 at least $i$ occurrences of $a$ and at least
 $j$ occurrences of $b$, where the input processed thus far has length
 at least $i + j + k$.
 Thus, $\mathcal B_{i,j,k}$ reaches the accepting state
 $z_k$ exactly on input of length at least $i + j + k$
 that have at least $i$ occurrences of $a$
 and at least $j$ occurrences of $b$.

 The cardinality of $Q$ is
 \[
  f(i,j,k) = (i+1)\cdot (j+1) + k\cdot j + k\cdot i + k
 \]
 (taking into account that the first two sets in the union defining $Q$
 have elements in common).
 In order to get an upper bound for the state complexity of $\perm(L(\mathcal A))$
 as a function of the size of $\mathcal A$,
 we determine for which values of $i,j,k$, where
 $i + j + k = n - 1$, the function $f(i,j,k)$ has a maximal value.
 The function $f$ is maximized if $ij + kj + ki$ is maximal,
 thus if $i = j = k = \frac{n-1}{3}$. More generally,
 \[
  \max_{i+j+k=n-1} f(i,j,k) = \left\{ % todo prüfen
   \begin{array}{ll}
  \frac{n^2 + n + 1}{3} & \mbox{if } n \equiv 1 \pmod{3}; \\
  \frac{n^2 + n}{3}     & \mbox{otherwise.}
  \end{array}
  \right.
 \]
 In~\cite[Lemma 4.2]{DBLP:journals/tcs/ChoGHKPS17}, as every chain automaton
 is a general chain automaton recognizing a finite language,
 it was shown that for $n \equiv 1 \pmod{3}$ 
 there exists a language recognized by a chain automaton
 with $n$ states such that any automaton
 for the commutative closure needs at least $\frac{n^2 + n + 1}{3}$
 many states.~\qed
\end{proof}

%% file: complexity.tex
Here, we consider the alphabet to be fixed in advance and not part of the input.

In model checking, when the specification and the implementation could be represented by finite automata,
the inclusion problem arises naturally~\cite{DBLP:conf/concur/AbdullaJNS04,DBLP:conf/banff/Vardi95}.
In this problem, we are given two automata and ask if the recognized language of the first is contained
in the recognized language of the second automaton.
In~\cite{DBLP:journals/iandc/BouajjaniMT07} it was shown that the universality problem, i.e., deciding if a given APC\footnote{Or a partially ordered NFA, which follows by Lemma~\ref{lem:NFA_to_APC_P}.}
denotes $\Sigma^*$, is $\PSPACE$-complete, \emph{even for fixed binary alphabets}. This implies $\PSPACE$-completeness
of the inclusion problem.

Here, we show the somewhat surprising result that the above decision problems
are polynomial time solvable modulo permutational equivalence, i.e., if we ask
the same questions for the commutative closure of the input languages,
see Theorem~\ref{thm:inclusion_perm} and Corollary~\ref{cor:universality_perm}.

This result is not as artificial as it might seem. For example, consider
the introductory example from regular model checking in Section~\ref{sec:introduction}.
Here, the set of reachable configurations $0^*10^*$
is closed under the commutative closure,
as well as the set of bad configurations $(0+1)^*1(0+1)^*1(0+1)^*$
and its complement.
More specifically, these sets are commutative languages and the original decision problem 
is equivalent to the same decision problem modulo permutational equivalence.

At the heart of this result lies the fact that the PDFA constructed in the proof
of Theorem~\ref{thm:apc_bound} could be constructed, for a fixed alphabet, in polynomial time.
This will be shown in Proposition~\ref{prop:NFA_to_PDFA_for_perm_in_P}.
But before this result, let us first state that, with respect
to polynomial time, it makes no difference if the input is given as an APC expression
or a partially ordered NFA.

\begin{lemmarep} 
\label{lem:NFA_to_APC_P}
 For a given partially ordered NFA $\mathcal A$
 an APC expression of $L(\mathcal A)$
 could be computed in $\PTIME$ and for every APC expression
 a partially ordered NFA is computable in $\PTIME$.
 This result also holds for variable input alphabets.
\end{lemmarep}
\begin{proof}
 Kleene's algorithm~\cite{HopUll79} for the conversion of any NFA\footnote{In~\cite{HopUll79} it is only described for deterministic automata, but it actually works for NFAs the same way.}
 to a regular expression could
 be used, arguing that it yields an APC when applied to a partially ordered NFA.
 However, we also give a more direct approach.
 Let $\mathcal A = (\Sigma, Q, \delta, q_0, F)$.
 As in the proof of Proposition~\ref{prop:NFA_to_PDFA_for_perm_in_P},
 we can assume $q_0$ is a minimal state with respect to the partial order
 of the automaton.
 Compute a topological ordering $Q = \{q_0, q_1, \ldots, q_n\}$.
 For $i \in \{0,\ldots,n\}$, let $L_i = L((\Sigma, Q, \delta, q_i, F))$
 be the language when $\mathcal A$ is started at state $q_i$.
 %By assumption, and as from $q_n$ no transitions goes out to another state,
 %we have $L_n = \Sigma_n^*$ with $\Sigma_n = \{ a \in \Sigma \mid q_n \in \delta(q_n, a) \}$.
 For any other $i \in \{0, \ldots, n\}$
 set $\Sigma_i = \{ a \in \Sigma \mid q_i \in \delta(q_i, a) \}$.
 We have
 \[ % diese notation bzw regular expression wird vorausgesetzt noch schreiben, todo
  L_i = \bigcup_{\substack{a \in \Sigma, q_j \in Q \\ q_j \in \delta(q_i, a) \setminus \{q_i \}}} 
   \Sigma_i^* \cdot \{ a \} \cdot L_j.
 \]
 This equation is actually valid for any NFA.
 However, as $\mathcal A$ is partially ordered with $q_0 < q_1 < \ldots < q_n$,
 in the above equation we have, for $q_j \in \delta(q_i, a) \setminus \{q_i\}$
 with $q_j \in Q$ and $a \in \Sigma$, that $i < j$.
 So, the index works as an induction parameter, where the base cases
 are the maximal indices respectively states.
 For the maximal states $q_j$ we must have $L_j = \Sigma_j^*$,
 in particular $L_n = \Sigma_n^*$.
 Hence, we can assume inductively that the $L_j$
 are given by an APC, which yields the claim for $L_i$
 and ultimately for $L_0 = L(\mathcal A)$ (recall the characterization of APCs, mentioned in Section~\ref{sec:preliminaries},
 as the closure of $\Gamma^*$, $\Gamma \subseteq \Sigma$, and $\{a\}$, $a \in \Sigma$, under concatenation
 and finite union).
 It is clear that the above recursive computation could be performed in polynomial time
 (in fact, the topological ordering does not need to be computed and was only
 used for notational convenience).
 Also, note that by omitting $\Sigma_i^*$ % todo, damit einfacherer beweis in P-algorithmus zur konstruktion von PDFA??
 in the above equation, the language $L^{\operatorname{simple}}(\mathcal A)$
 could be computed by the resulting recursion.
 
 Conversely, suppose we have an APC expression $L$ over $\Sigma^*$.
 Now, a partially ordered NFA
 for an expression of the form $\Sigma_0^* a_1 \Sigma_1^* \cdots a_n \Sigma_n^*$
 is a single path labelled by $a_0 \cdots a_n$
 and with self-loops labelled by the symbols in $\Sigma_i$ at the $i$-th state in this path.
 A finite union of such languages corresponds to an NFA
 that branches in the initial state to a path corresponding
 to each part of the form $\Sigma_0^* a_1 \Sigma_1^* \cdots a_n \Sigma_n^*$
 in the union.

 Lastly, we see that the presented algorithms also run in polynomial time when
 the alphabet is allowed to be part of the input.~\qed
\end{proof}

So, we are ready to derive that from a given partially ordered NFA,
a PDFA recognizing the commutative closure could be computed in $\PTIME$.

\begin{propositionrep}
\label{prop:NFA_to_PDFA_for_perm_in_P}
 Given a partially ordered NFA $\mathcal A$ with $n$ states,
 the recognizing PDFA for $\perm(L(\mathcal A))$
 from Theorem~\ref{thm:apc_bound}
 could be constructed in polynomial time for a fixed alphabet.
 More precisely in time $O(n^{|\Sigma|+2})$.
\end{propositionrep}
\begin{proofsketch} This is only a rough and intuitive outline of the procedure.

 Let $\Sigma = \{a_1, \ldots, a_k\}$ and $\mathcal A = (\Sigma, S, \mu, s_0, E)$
 be a partially ordered NFA. We outline a polynomial time
 algorithm to compute $\mathcal B = (\Sigma, Q, \delta, q_0, F)$ as
 defined in the proof of Theorem~\ref{thm:apc_bound}.
 We can assume that $s_0$ is minimal for the partial order of $\mathcal A$
 and every maximal state is final.
 Set $n_a = \max\{ |u|_a : u \in L^{\operatorname{simple}}(\mathcal A) \}$
 for $a \in \Sigma$

 The state set, and hence the numbers $n_a$, could be computed by a dynamic programming scheme starting
 at the maximal final states and ending at the start state. For each letter $a \in \Sigma$,
 we store at every state $q$ the number $\max\{ |u|_a \mid \delta(q, u) \in F \mbox{ and no loops are entered by $u$ in $\mathcal A$} \}$, i.e.,
 the longest unary projection string for that letter when starting at this state,
 ending at a final state and traversing no self-loops\footnote{So, essentially we are working
 in the automaton that results if we delete all self-loops, which gives a recognizing
 automaton for $L^{\operatorname{simple}}(\mathcal A)$ for partially ordered NFAs $\mathcal A$.}.
 For a final maximal state, those numbers are initialized to zero
 and for every other state, they are computable from the predecessor states.
 For the start state, the last state in this procedure, these are precisely the numbers
 $n_a$, from which $Q$
 is easily constructible.
 
 % todo im appendix auch mit s_j bezeichnen, nicht i_j
 %Then, we have to compute the transition function and the final state set.
 The computation of the transition function and the final state set is more involved.
 Note that for states $(s_1, \ldots, s_k) \in Q$
 with $s_i < n_{a_i}$ for $i \in \{1,\ldots,k\}$
 the transition function is easily computable. The only difficulty is to determine which
 ``boundary'' states should be labeled by self-loops.
 We do this by constructing an auxiliary automaton~$\mathcal A'$ out of~$\mathcal A$ by ``unfolding'' the self-loops into paths of length $|S| + 1$.
 The automaton~$\mathcal A'$ then has no loops anymore.
 Now, we label the states of this auxiliary automaton
 with those states from $Q$ that are reachable
 in~$\mathcal B$ by words that go from the start state to the state
 under consideration of~$\mathcal A'$.
 If such a word passes an unfolded path completely, then, as they
 are sufficiently long, we know that it must traverse a self-loop
 in $\mathcal A$ labeled by the same letter~$a$ as the unfolded path. In this case, for every ``boundary'' state of $Q$ in the labeling
 of the target state of the word in~$\mathcal A'$
 we add a self-loop for the letter $a$ to that state from $Q$ in $\mathcal B$.
 
 Finally, a state from $Q$
 is declared to be final if and only if it appears in a label of a final state of $\mathcal A'$.

 This procedure indeed computes $\mathcal B$ and could be made to run
 in the stated time bound.~\qed
\begin{comment}
 This ensures that any words read by this automaton that traverses such a sufficiently
 long path has to traverse a loop in $\mathcal A$.

 To do this, we use an auxiliary automaton $\mathcal A'$ without loops, constructed
 out of $\mathcal A$ by ``unfolded'' the self-loops of $\mathcal A$
 into paths of length at most $|S| + 1$ (i.e. the automaton accepts
 those words that traverse the loops in $\mathcal A$ at most $|S|$ times).
 The unfolded loops give paths in $\mathcal A'$ that ensure, if a word traverses such a path,
 it has to necessarily traverse a loop labeled by the same letter in $\mathcal A$.
 This information is then used to deduce when we have to add a self-loop 
 at a ``boundary'' state in $Q$ corresponding to that letter, as
 the defining condition to add such a loop to $\mathcal B$
 takes place precisely if the letter in that direction labels a self-loop
 in $\mathcal A$ along a path corresponding to a word leading to that state in $\mathcal B$.
 We then label the states with states from $Q$, where a state from $Q$
 is declared to be final if and only if it appears in a label of a final state of $\mathcal A'$
 and when a boundary state of $Q$ labels a state of $\mathcal A$ that enforces
 a loop in $\mathcal A$, we add a self-loop to that state of $Q$.

 This procedure indeed computes $\mathcal B$ and could be made to run
 in the stated time bound.~\qed
\end{comment}
\end{proofsketch}
\begin{proof}
 Let $\Sigma = \{a_1, \ldots, a_k\}$
 and $\mathcal A = (\Sigma, S, \mu, s_0, E)$ be the input NFA.
 We refer to the proof of Theorem~\ref{thm:apc_bound}
 for the definition of $\mathcal B = (\Sigma, Q, \delta, q_0, F)$ and further notation
 and show that the defining parameters
 of the automaton $\mathcal B = (\Sigma, Q, \delta, q_0, F)$
 could be computed in polynomial time by giving 
 algorithms running in polyomial time.
 
 \begin{enumerate}
 \item The state set $Q$:
  % dynamic programming, vektor (i_1, ..., i_k)
  % mit maximmalen elementen. [kann annehmen, dass maximale final, sonst weglassen, innere final egal]
  We use that the states are partially ordered by the reachability relation.
  If a state is maximal for this order and not final, we can ignore this
  state without altering the recognized language.
  So, we assume that every maximal state is final.
  Then, for every state we compute a vector $(i_1,\ldots, i_k)$
  that stores for any $j \in \{1,\ldots,k\}$
  in the entry $i_j$ the maximum of the number of times the letter $a_j$ appears 
  in the words
  that go from this state to a maximal final state
  without repeating a state, i.e., not entering a loop.
  We compute these vectors according to the following rules:
  \begin{enumerate}
  \item For any maximal final state, set it equal to $(0,\ldots,0)$.
  
  \item For any $p \in Q$ that is not maximal compute $(i_1, \ldots, i_k)$
   according to the following scheme.
   Inductively, we can assume that the vectors for states strictly larger
   in the partial order are already computed.
   We use the vectors at the direct successors (i.e. those reachable by a single letter)
   to compute the vector for the state under consideration.
   For every state $q \in Q$ reachable by at least one letter from $p$,
   let $\Sigma_q \subseteq \Sigma$ be the set of letters from which we can reach it,
   i.e., $\Sigma_q = \{ a \in \Sigma \mid \delta(p, a) = q \}$.
   By choice of $q$, we have $\Sigma_q \ne \emptyset$.
   Let $(i_1, \ldots, i_k)$ be the vector corresponding to $q$.
   Then, let $(i_1', \ldots, i_k')$
   be the vector defined by
   \[
    i_j' = \left\{ 
    \begin{array}{ll}
     i_j + 1 & \mbox{if } a_j \in \Sigma_q; \\
     i_j     & \mbox{otherwise.}
    \end{array}
    \right.
   \]
   Let $S_p \subseteq \mathbb N_0^k$ be the set of all these
   vectors for each state $q \in Q$ as above.
   Then, the vector of the state $p \in Q$ is
   \[
    (m_1, \ldots, m_k)
   \]
   where $m_j = \max\{ i_j \mid (i_1, \ldots, i_k) \in S_p \}$ for $j \in \{1,\ldots,k\}$.

   Finally, for the state $s_0$, which could be assumed to be minimal, for otherwise
   all predecessor states could be ignored without altering the recognized language,
   we have $m_j = \max\{ |u|_{a_j} : u \in L^{\operatorname{simple}(\mathcal A)}\}$.
   Recall that in the proof of Theorem~\ref{thm:apc_bound},
   these number were denoted by $n_j$, i.e., the vector $(m_1, \ldots, m_k)$
   for the state $s_0$ is precisely the vector $(n_1, \ldots, n_k)$
   introduced in the proof of Theorem~\ref{thm:apc_bound}.
   Having these numbers, the state set $Q$ could be easily constructed,
   as it is only a cartesian product $[n_1 + 1] \times\ldots\times [n_k+1]$.
   
   \medskip

   Note that, in this computation, possible final states that are not maximal
   are handled like ordinary non-maximal states in the second step above.
   
  \end{enumerate}
  
 \item The transition function and the final state set $F$.    
   % permbedignunge prüfen?
  % nur prüfen ob auf pfad selflopp mit dem a_j
  % auftaucht.

  % noch icht klar.
  %
  % l simple klar, dag und dann nur pfade auf finalzustände, dabei noch selfloops sammeln
  % für größere wörter [evtl self-loops "auffolden" zu |Q|]
  
  % dann von q_0 buchstaben zu allen zuständen einlesen und parallel
  % schritt machen in produkt B, mrkieren, wenn man auf final stößt,
  % dafür auch weider vektor an zuständen der für position steht speichern [oder menge
  %
  % ^|Q| (this choice is very "liberal")
  %
  % A_unfold function von zuständen [implementes as hash-table] nach 2^{Q}
  % dann auch "rekursionsschema" mit vorgängern
  
  Recall from the proof of Theorem~\ref{thm:apc_bound}
  that $Q = [n_1 + 1] \times \ldots \times [n_k + 1]$
  with $n_j = \max\{ |u|_{a_j} \mid u \in L^{\operatorname{simple}}(\mathcal A) \} + 1$.
  The transition function among states $(s_1, \ldots, s_k) \in Q$
  with $s_j < n_j$, $j \in \{1,\ldots,k\}$, when an $a_j \in \Sigma$
  is read is easily computable -- simply increase $s_j$ by one.
  Hence, the only difficulty is to determine if a self-loop should be
  added in case $s_j = n_j$.
  Intuitively, we ``unfold'' the self-loops in $\mathcal A$
  into long enough paths of length $|S|$,
  such that we can detect if an input word has to pass a self-loop, i.e.,
  we read along these ``unfolded'' paths which after a certain
  point, as they have length $|S|$, enforce a loop in $\mathcal A$. 
  We do this by maintaining for each state of the unfolded automaton
  a list of states of $\mathcal B$ that are reachable by words that go into
  this state of the unfolded automaton.

  % das einmal am anfang schreiben ? todo
  As before, we can assume that from every state of $\mathcal A$,
  a final state is reachable and that $s_0$ is the minimal element
  in the partial order induced by the reachability relation.
  
  First, we construct an automaton $\mathcal A' = (\Sigma, S', \mu', s_0', E')$
  by ``unfolding'' the self-loops. The result $\mathcal A'$ will have no loops.
  For a state $s \in S$, set $\Sigma_s = \{ a \in \Sigma \mid s \in \mu(s, a) \}$.
  If $\Sigma_s \ne \emptyset$, do the following: %relationschreibweise für funktion in mengen todo fussnote
  (1) remove the self-loops, % todo genauer mu' so definieren.
  %
  % entscheidende eigenschaft warum loops zu |Q| unfolded, dass max \{ |u|_a a \in Sigma, u \in L }
  % durch |Q| begrenzt, hätte auch das nehmen können, fussnote.
  (2) set $s_1 = s$ and add the new states $\{ s_2, \ldots, s_{|S|+1} \}$
   and the following transitions
  \begin{multline*}
   \{ (s_{i-1}, a, s_i) \mid a \in \Sigma_s, i \in \{2,\ldots,|S|+1\} \} \cup \\
   \{ (s_i, b, t) \mid b \in \Sigma, t \in \mu(s, b) \setminus\{s\}, i \in \{2,\ldots, |S|+1\} \}.
  \end{multline*}
  For every such state $s$, we add $|S|$ new states,
  hence $|S'| \in O(|S|^2)$.
  Also, set $E' = E \cup \{ s_i \mid s \in E, i \in \{2,\ldots,|S|+1\} \}$
  and $s_0' = s_0$.

  In the following computation, we will only refer to the automaton $\mathcal A'$.
  First, topologically sort the states of $\mathcal A'$. 
  Then, for any state $s \in S'$, we compute a subset $T_s \subseteq Q$ of states $Q$ from $\mathcal B$.
  The states in $T_s$ are precisely those states of $\mathcal B$ that 
  are reachable by words that go from the start state to $s$ in $\mathcal A'$.
  Note that, by construction of $\mathcal A'$, this is only a finite set of words.
  Also note that, as a subset is associated with every state, we
  will construct at most $|S'|$ many subsets during the procedure, i.e., 
  we have no blow-up here and the final algorithm will run in polynomial time.

  Initially, to every state in $S'$ the empty set is assigned.
  Then, proceed in the following way:
  \begin{enumerate}
  \item Begin with the start state $s_0'$ and associate the subset $T_{s_0} = \{ (0,\ldots, 0) \}$.
  
  \item For the current state $s \in S'$ under consideration, do the following.   
     For any direct successor state of $s$, i.e., those reachable
  by a single symbol, if $a_j$, $j \in \{1,\ldots,k\}$
  is the letter leading to this successor state, update
  the assigned subset of this successors state by adding
  to it the subset from the current state, but with every
  vector inside increased by the vector that has zero everywhere but
  at the $j$-th position, where the entry is one, to it, except for the states
  inside this set that equal $n_j$ at the $j$-th position.
  In this case, do nothing and add a self-loop labelled
  by $a_j$ to the state in $Q$ under consideration (here, that
  we can reach a final state is important, so that the word read up so far
  is a prefix of some word in $L$, see the definition of
  the transition function in the proof of Theorem~\ref{thm:apc_bound}
  for comparison, as this is a defining condition to add self-loops).
  More precisely, if $t \in S'$ is a successor state to $s \in S'$ for the letter $a_j$, i.e.,
  $t = \delta(s, a_j)$, then,
  add to $T_t$ the sets
  \begin{multline*}
    \{ (i_1, \ldots, i_{j-1}, i_j + 1, i_{j+1}, \ldots, i_k) \mid \\ 
     (i_1, \ldots, i_{j-1}, i_j, i_{j+1}, \ldots, i_k)  \in T_s \mbox{ and } i_j < n_j \} 
  \end{multline*}
  and
  \begin{multline*}
   \{ (i_1, \ldots, i_{j-1}, n_j, i_{j+1}, \ldots, i_k) \mid 
     (i_1, \ldots, i_{j-1}, n_j, i_{j+1}, \ldots, i_k)  \in T_s \}.
  \end{multline*}
  Also, in case the following holds true
  \[
  \{ (i_1, \ldots, i_{j-1}, n_j, i_{j+1}, \ldots, i_k) \mid 
     (i_1, \ldots, i_{j-1}, n_j, i_{j+1}, \ldots, i_k)  \in T_s \} \ne \emptyset,
  \]
  add the self-loop 
  \[
  \mu( (i_1, \ldots, i_{j-1}, n_j, i_{j+1}, \ldots, i_k) , a_j ) = (i_1, \ldots, i_{j-1}, n_j, i_{j+1}, \ldots, i_k)
  \]
  to $\mathcal B$.

  \item After we have done the previous step for every direct successor state, move on to the
  next state in the topological order.
  \end{enumerate}

  In total, in the above procedure, only a polynomial number of steps
  are performed.
  So, also the subsets involved in the computation
  do not grow too much.

  Finally, for every final state, collect the associated subsets of $Q$
  and these states are the final states of $\mathcal B$.
  More precisely, $(i_1, \ldots, i_k) \in F$ if and only if
  there exists $s \in E'$ such that
  $
   (i_1, \ldots, i_k) \in T_s.
  $

 \end{enumerate}
 We have $|Q| \le n^{|\Sigma|}$, and in the above procedure of unfolding and traversing
 the at most $n^2$ states of $\mathcal A$, we have to perform, for each letter
 at most $|Q|$ many operations, as this is the maximal size of the associated sets.
 So, we have given a polynomial time algorithm to compute $\mathcal B$ in time $O(|\Sigma||Q|^{|\Sigma|+2}) = O(|Q|^{|\Sigma|+2})$.~\qed
\end{proof}

% erwähnen bzw lemma eingabe als NFA oder expression egal, l^simple

 With Proposition~\ref{prop:NFA_to_PDFA_for_perm_in_P}, 
 we derive that, given two APCs, the inclusion 
 problem modulo permutational equivalence
 is solvable in polynomial time.

\begin{theoremrep}\label{thm:inclusion_perm}
 Fix an alphabet $\Sigma$. Then, the following problem is in $\PTIME$:
%\begin{quote} 
 \\
 \emph{Input:} Two APC expressions $L_1, L_2$ over $\Sigma^*$. \\ 
 \emph{Question:} Is $\perm(L_1) \subseteq \perm(L_2)$?
%\end{quote}
\end{theoremrep}
\begin{proof}
 First, we construct two NFAs $\mathcal A_1$ and $\mathcal A_2$
 for $L_1$ and $L_2$, which could be done in $\PTIME$ by Lemma~\ref{lem:NFA_to_APC_P}.
 Then, we compute PDFAs $\mathcal B_1$
 and $\mathcal B_2$ for their respective commutative closures,
 which could be done in $\PTIME$ by Proposition~\ref{prop:NFA_to_PDFA_for_perm_in_P}.
 Now, for deterministic automata 
 the inclusion probelm is solvable in $\PTIME$.
 More precisely, we have:
 $
  L(\mathcal B_1) \subseteq L(\mathcal B_2)
  \Leftrightarrow L(\mathcal B_1) \cap \overline{L(\mathcal B_2)} = \emptyset.
 $
 Then, on deterministic automata the Boolean operations are performable in polynomial
 time with the product automaton construction~\cite{HopUll79}
 and switching of final and non-final states. 
 Also, the non-emptiness problem, i.e., deciding if the recognized language
 is non-empty, for automata (even NFAs) is $\NL$-complete~\cite{DBLP:journals/iandc/HolzerK11},
 hence also in $\PTIME$. So, we can perform the above emptiness check in~$\PTIME$.~\qed
\end{proof}

 Given an APC, the universality problem modulo permutational equivalence
 is solvable in polynomial time, as it is reducible
 to the corresponding inclusion problem up to permutational equivalence.
 
\begin{corollary}\label{cor:universality_perm}
 Fix an alphabet $\Sigma$. Then, the following problem is in $\PTIME$:
%\begin{quote} 
 \\
 \emph{Input:} An APC expression $L$ over $\Sigma^*$. \\ 
 \emph{Question:} Is $\perm(L) = \Sigma^*$?
%\end{quote}
\end{corollary}

As for commutative languages $L \subseteq \Sigma^*$
we have $\perm(L) = L$, we get the next corollary. This generalizes a corresponding reduction
of complexity for unary alphabets~\cite{DBLP:journals/iandc/KrotzschMT17}.

\begin{corollary}
\label{cor:com_apc_inclusion}
 Fix an alphabet $\Sigma$. Given an APC describing a commutative language, the universality problem is in $\PTIME$.
 Also, given two APCs describing commutative languages,
 the inclusion problem is solvable in polynomial time. 
\end{corollary}

% \begin{corollary} % verallgeminert unören fall masopust [referenz]
% \label{cor:com_apc_universality}
%  Fix an alphabet $\Sigma$. Given an APC describing a commutative language,
%  the universality problem is solvable in polynomial time.
% \end{corollary}